\documentclass[pra,onecolumn]{revtex4}
\usepackage{amsthm}
\usepackage{amsmath}
\usepackage{mathtools}
\usepackage{latexsym}
\usepackage{amsfonts}
\usepackage{amssymb}
\usepackage{color}
\usepackage{bbm,dsfont}
\usepackage{graphicx}
\usepackage{hyperref}
\usepackage{cleveref}
\usepackage{enumerate}
\usepackage[caption=false,subrefformat=parens,labelformat=parens,position=top]{subfig}
\usepackage{verbatim}
\usepackage{tcolorbox}

\newtheorem{proposition}{Proposition}
\newtheorem{proposition?}{Proposition?}
\newtheorem{theorem}{Theorem}
\newtheorem{lemma}{Lemma}

\theoremstyle{definition}

\newcommand{\C}{\mathbb C} 
\newcommand{\half}{\frac{1}{2}} 
\newcommand{\abs}[1]{\left| #1 \right|} 

\newcommand{\id}{I} 
\newcommand{\hi}{\mathcal{H}} 
\newcommand{\lh}{\mathcal{L(H)}} 
\newcommand{\trh}{\mathcal{T(H)}} 
\newcommand{\sh}{\mathcal{S(H)}} 
\newcommand{\eh}{\mathcal{E(H)}} 
\newcommand{\ip}[2]{\left\langle #1 | #2\right\rangle} 
\newcommand{\kb}[2]{|#1\rangle\!\langle#2|} 
\newcommand{\proj}[1]{\kb{#1}{#1}}
\newcommand{\tr}[2][]{\mathrm{tr}_{#1}\!\left[#2\right]} 
\newcommand{\comm}[2]{\left[#1,#2\right]} 

\newcommand{\Ao}{\mathsf{A}}
\newcommand{\Bo}{\mathsf{B}}
\newcommand{\Do}{\mathsf{D}}
\newcommand{\To}{\mathsf{T}}
\newcommand{\Xo}{\mathsf{X}}

\newcommand{\In}{\mathcal{I}} 

\newcommand{\bx}{\boldsymbol{x}}

\begin{document}
\title{Quantum State Discrimination via Repeated Measurements and the Rule of Three}

\author{Tom Bullock}

\email{tombul@utu.fi}
\affiliation{QTF Centre of Excellence, Turku Centre for Quantum Physics, Department of Physics and Astronomy,
University of Turku, FI-20014 Turun yliopisto, Finland
}

\author{Teiko Heinosaari}
\affiliation{QTF Centre of Excellence, Turku Centre for Quantum Physics, Department of Physics and Astronomy,
University of Turku, FI-20014 Turun yliopisto, Finland
}

\begin{abstract}
    The task of state discrimination for a set of mutually orthogonal pure states is trivial if one has access to the corresponding sharp (projection-valued) measurement, but what if we are restricted to an unsharp measurement?
    Given that any realistic measurement device will be subject to some noise, such a problem is worth considering.
    In this paper we consider minimum error state discrimination for mutually orthogonal states with a noisy measurement.
    We  show that by considering repetitions of commutative L\"uders measurements on the same system we are able to increase the probability of successfully distinguishing states.
    In the case of binary L\"uders measurements we provide a full characterisation of the success probabilities for any number of repetitions.
    This leads us to identify a `rule of three', where no change in probability is obtained from a second measurement but there is noticeable improvement after a third.
    We also provide partial results for $N$-valued commutative measurements where the rule of three remains, but the general pattern present in binary measurements is no longer satisfied.
\end{abstract}

\maketitle

\section{Introduction}\label{sec:int}

The task of discriminating states via `one-shot' measurements is fundamental to a number of processes in many scenarios: classically, the solution of a calculation performed on a computer cannot be read, nor can a private encryption key be shared, unless we are able to correctly identify the resulting bits.
This remains true within the quantum realm, where any quantum process is only ever as good as our ability to distinguish the possible states the system may be in.
However, due to the nature of states within quantum mechanics, we are subject to difficulties not seen classically: whereas a single measurement is generally sufficient to fully distinguish between a set of possible pure states in a classic system, for a collection of pure quantum states it is impossible to distinguish them in such a setting, with the exception of a set of mutually orthogonal states.
As such, the aim is to maximise our chances of distinguishing the states via clever choices of measurement. 

As is the case with a number of notions within quantum theory, there is no clear-cut `best choice' for what one wishes to optimise over in a state discrimination task, and by extension no obvious choice of measurement to perform in any given circumstance.
In the 1970's Helstrom \cite{QDET76}, Holevo \cite{Holevo73,PSAQT82}, and Yuen \emph{et al.} \cite{YuKeLa75} began work on minimum error state discrimination, which aimed to minimise the average error in determining each possible state in a set via their corresponding measurement outcome.
In the late 1980's Ivanonic \cite{Ivanovic87} (followed by Dieks \cite{Dieks88} and Peres \cite{Peres88}) considered so-called unambiguous state discrimination, where the goal was no longer to minimise error, but rather to maximise the chance of unambiguously distinguishing states where possible. 
More recently an additional figure of merit referred to as ``maximum confidence'' was introduced by Kosut \emph{et al.} \cite{KoWaElRa04} and Croke \emph{et al.} \cite{CrAnBaGiJe06} (who coined the term), which refers to how confident we may be that we started with a particular state, given that we obtained its corresponding measurement outcome.
Each of these figures of merit come with shortcomings: in minimum error state discrimination we will still find some instances of incorrect measurement outcomes for a given state, whereas the latter two admit inconclusive measurement results where no state can be inferred.
As such, one generally decides a particular figure of merit based on their preferred trade-off.
For more in-depth summaries of state discrimination tasks and their applications we refer the reader to \cite{BaCr09,Bergou10,BaKw15}.

It is well known that for distinguishing a set of mutually orthogonal pure states (with regard to any figure of merit) the obvious choice is a sharp observable containing the states in question as projections.
However, in this work we consider minimum error state discrimination in the case where one is unable to use such a sharp observable, and instead some unsharp variant is available.
This construct is in contrast to the standard problem presented, where a set of non-orthogonal states are distinguished, but considering our scenario is physically justified: unsharp measurements can allow for the possibility of sequential or repeated measurements without complete loss of the initial state, and furthermore encapsulate the reality that any measurement will contain some inherent noise, be it the result of mechanical or human error.
In fact, it is impossible to perform projective  measurements on quantum systems  using  finite resources \cite{GuFrHu20}.
While this may seem to be a hindrance, we shall see that the aforementioned admission of repeated measurements works in our favour.

In \cite{HaHeKu16} it was shown that in the limit of number of repetitions going to infinity, the information content of an unsharp binary observable tends to that of its spectral counterpart, but no consideration was given to intermediate steps. 
This is the point of interest for this paper, as it is not only a more realistic scenario, but allows us to see where the trade-off point lies where an increased number of measurements provides diminishing returns in terms of increased success probability.
In a realistic scenario one would aim to perform a finite number of measurement rounds and in doing so reach a sufficiently high confidence about the initial state.

The use of repeated measurements in state discrimination tasks have previously been adopted in the context of unambiguous state discrimination \cite{BeFeHi13,PaZhXuLiCh13,FiHaHiBe20}, and for minimum error discrimination \cite{RoPaMaGi17,CrBaWe17}.
However, in the studied scenarios the observables measured differ over the course of the measurement process (with, for example \cite{RoPaMaGi17} relying on sequential binary measurements to approximate an $N$-valued observable as described in \cite{AnOi08}).
In another example \cite{BeFeHi13} there are multiple agents aiming to gather the same information, solely based on their individually obtained measurement outcomes.
In this paper we instead restrict ourselves to performing repetitions of the same $N$-valued observable on the system, thereby receiving a string of outcomes from which we can post-process the data to determine the most likely state among $N$ alternatives. 
Mathematically, this translates to constructing a joint observable whose effects we partition to form a new $N$-valued observable, which is then used to calculate the success probability.

We begin this paper with an overview of the necessary concepts in Section \ref{sec:prelim}; these contain the mathematical descriptions of finite outcome observables and repeated measurements, as well as expressing the problem of minimum error state discrimination and post-processing of observables.
We motivate our main work in Section \ref{sec:motiv}, where we explicitly calculate the success probability for distinguishing two eigenvectors of the $x$-direction operator $\sigma_x$ by repeated measurements of the unsharp spin-$x$ observable $\Xo_t$. 
This example introduces us to the ``rule of three'' that will be prevalent in the cases considered later: the second measurement does not lead to an increase in the success probability, but the third does.
The case of binary observables is considered in Section \ref{sec:main} for an arbitrary number of repetitions, and we find a general pattern whereby the success probability only increases for odd numbers of repetitions.
We extend the results of Section \ref{sec:main} by providing some partial results for commutative $N$-valued observables ($N>2$) in Section \ref{sec:higher}, where the rule of three is found to still hold, though the general pattern for binary observables is no longer seen to hold.
Finally, we discuss our results in Section \ref{sec:conc}.

\section{General framework for state discrimination via repeated measurements }\label{sec:prelim}

\subsection{Quantum measurements}\label{subsec:seqmeas}

In this section we recall the basic formalism of quantum measurements.
(For a more detailed discussion see, for example \cite{QM16}.)
We shall generally consider an arbitrary separable Hilbert space $\hi$ of (possibly infinite) dimension $d$ and its space $\lh$ of bounded linear operators.
The operators we will primarily work with are the so-called effects on $\hi$: $\eh = \{E\in\lh\ |\ 0\leq E \leq I\}$, where the partial order is in terms of the expectation values of vectors of $\hi$, i.e., $E\leq F$ means that $\ip{\psi}{E\psi}\leq \ip{\psi}{F\psi}$ for all $\psi\in\hi$.
The state space of $\hi$, denoted by $\sh$, consists of positive operators of unit trace: $\sh= \{\rho\in\trh\ |\ \rho\geq0\ \mathrm{and}\ \tr{\rho}=1 \}$ where $\trh\subset \lh$ denotes the set of trace-class operators on $\hi$.

Quantum observables are mathematically described by positive operator-valued measures (POVMs), and we will restrict ourselves to $N$-valued observables, where $N$ is a finite integer.
An $N$-valued POVM $\Ao$ is a map from the finite outcome set $\Omega_\Ao=\{1,2,\dots,N\}$ to the set of effects; specifically, $\Ao:x \mapsto \Ao(x)\in\eh$ with $\sum_x \Ao(x)=\id$.
For a system prepared in a state $\rho$, the probability of obtaining the outcome $x$ when measuring observable $\Ao$ is given by the Born rule $p_\rho^\Ao(x) =\tr{\rho \Ao(x)}$.

An observable describes the measurement outcome statistics, but leaves open how the quantum state is transformed during the measurement process; this task is accomplished by an instrument, a mathematical object that describes both.
For a finite set $\Omega=\{1,\dots,N\}$, an instrument $\In$ is a map from $\Omega$ to the set of operations on $\trh$; that is, $\In:x \mapsto \In_{x}$ where 
\begin{enumerate}[(i)]
    \item $\In_{x}:\trh \rightarrow \trh$ is completely positive and linear for all $x\in\Omega$;
    \item $\In_{x}$ is trace non-increasing for all $x\in\Omega$;
    \item $\sum_x \In_{x}$ is trace preserving.
\end{enumerate}
The dual instrument $\In^*:\lh \rightarrow \lh$ defines a unique observable $\Ao_\In:\Omega\rightarrow \eh$ via $\Ao_\In (x) = \In^*_{x}(I)$ and thus $\tr{\In_{x}(\rho)}= p_\rho^{\Ao_\In}(x)$.
For any observable $\Ao$ there exist infinitely many instruments $\In$ such that $\Ao_\In=\Ao$, each describing some particular way of measuring it.
If the observable $\Ao$ is measured in a way described by the instrument $\In$ on a system in state $\rho$ and outcome $x$ is recorded, then the unnormalised state of the system after the measurement will be given by $\In_{x}(\rho)$.

\begin{figure}[t]
    \centering
    \includegraphics[width=0.7\textwidth]{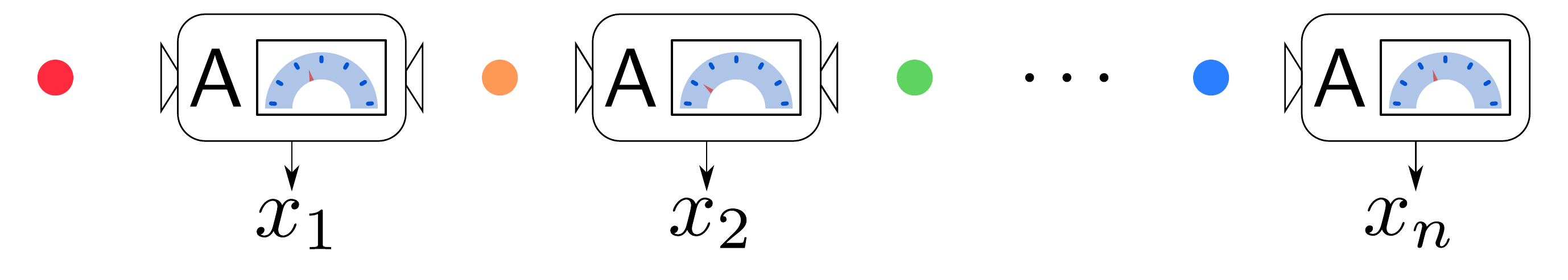}
    \caption{\label{fig:repeat} We consider the situation where we subject the system of interest to a repeated measurement of an observable $\Ao$ $n$ times. 
    After each measurement we record the result $x_i$ and the system experiences a state change $\rho\mapsto \In_{x_i}(\rho)$.
    Once each measurement is performed we are left with an $n$-tuple $(x_1,\dots,x_n)$ of measurement results, which can be considered the results of an $n$-round observable $\Ao^{(n)}_{\In}$.}
\end{figure}

\subsection{Repeated quantum measurements}\label{subsec:repmeas}

In our investigation, we restrict ourselves to situations where the same measurement apparatus is used repeatedly on the same system; see Figure \ref{fig:repeat}.
Suppose that $\In$ is the instrument describing the measurement of an observable $\Ao$.
Then, after $n$ repeated measurements we get an outcome array $(x_1,\ldots,x_n)\in\Omega_\Ao^n$ with probability
\begin{align}
\tr{\rho\Ao^{(n)}_{\In}(x_1,\ldots,x_n)}\coloneqq\tr{(\In_{x_n}\circ \cdots \circ \In_{x_1})(\rho)} \, ,
\end{align}
where $\In_{x_n}\circ \cdots \circ \In_{x_1}$ is the functional composition of the corresponding operations.
This equation, assumed to be valid for all input states $\rho$, determines the observable $\Ao^{(n)}_{\In}$, which we shall refer to as an $n$-round observable. 
The mathematical structure of $\Ao^{(n)}_{\In}$ depends crucially on the specific form of $\In$.
For instance, if $\In$ is of the measure-and-prepare form, then no subsequent measurement can give anything more than was already found in the first measurement. 
Further examples of repeated measurements that do not provide new information after the first, or some finite number of, repetitions are given in \cite{HaHeKu16}.

In the current work we limit ourselves to L\"uders instruments, i.e., we choose $\In_{x}(\rho) = \sqrt{\Ao(x)}\rho \sqrt{\Ao(x)}$.
For ease of notation we denote the corresponding $n$-round observable $\Ao^{(n)}_\In$ simply by $\Ao^{(n)}$ (since no other type of instrument will be considered in what follows).
Hence,
\begin{align}
\Ao^{(n)}(x_1,\ldots,x_n)=\sqrt{\Ao(x_1)}\cdots\sqrt{\Ao(x_{n-1})}\Ao(x_n)\sqrt{\Ao(x_{n-1})}\cdots\sqrt{\Ao(x_1)} \, .
\end{align}
If $\Ao$ is commutative (i.e. $\comm{\Ao(x)}{\Ao(y)}=0$ for all $x,y$), then the previous expression reduces to
\begin{equation}\label{eq:commrepmeas}
 \Ao^{(n)}(x_1,\ldots,x_n)=\Ao(x_1)\Ao(x_2)\cdots\Ao(x_n) \, .
\end{equation}
We observe that in this case the probability of getting a particular outcome array $(x_1,\ldots,x_n)$ in a state $\rho$ does not depend on the order of the outcomes.

\subsection{State discrimination in $n$ rounds}\label{subsec:statedisc}

We are considering a particular quantum information task, namely, the discrimination of states. 
We assume that the initial observable $\Ao$ with outcome set $\Omega_\Ao=\{1,\dots,N\}$ discriminates $N$ states $\{\rho_1, \dots,\rho_N\}$ with some error that is larger than in the optimal discrimination of these states.
The observable $\Ao$ can be, for instance, a noisy version of the observable that optimally discriminates the $N$ states.
The success probability $P_\mathrm{succ}^{(1)}$ for successfully distinguishing the states with $\Ao$ is given by
\begin{equation}
    P_\mathrm{succ}^{(1)}=\frac{1}{N}\sum_{j=1}^N \tr{\rho_j \Ao(j)},
\end{equation}
where we have assumed uniform a priori distribution of the states.

By repeating the measurement we hope to increase the probability of guessing the correct state.
After $n$ L\"uders measurements of $\Ao$, resulting in the observable $\Ao^{(n)}$ with outcome set $\Omega_\Ao^n$, we possess statistics for $N^n$ possible strings of measurement outcomes.
In order to assess the ability of $\Ao^{(n)}$ in distinguishing the original $N$ states, we need to post-process $\Ao^{(n)}$ such that we are left with an $N$-valued observable denoted by $\Bo^{(n)}$ with outcome set $\Omega_\Ao$.
This means that for each outcome array $(x_1,\ldots,x_n)$, we need to decide the most likely state $\rho_j$ and hence relabel this outcome array into $j$. 

Mathematically, the post-processing is performed by a Markov kernel $w:\Omega_\Ao \times \Omega_\Ao^n \rightarrow [0,1]$ satisfying $w(j, \boldsymbol{x}) \eqqcolon w^{j}_{\boldsymbol{x}} \geq 0$ for all $j\in\Omega_\Ao$ and $\boldsymbol{x}\in\Omega_\Ao^n$, and $\sum_{j} w^{j}_{\boldsymbol{x}} =1$ for all $\boldsymbol{x}\in\Omega_\Ao^n$.
The numerical value $w(j, \boldsymbol{x})$ is the probability that $\bx$ is relabeled into $j$. 
If each $\bx$ determines a unique $j$ (i.e., $w(j, \boldsymbol{x})\in \{0,1\}$ for all $j\in\Omega_\Ao$ and $\boldsymbol{x}\in\Omega_\Ao^n$), then we say that $w$ is deterministic.
The post-processed observable is given by 
\begin{equation}
    \Bo^{(n)}(j) = \sum_{\boldsymbol{x}\in\Omega_\Ao^{(n)}} w(j,\boldsymbol{x}) \Ao^{(n)}(\boldsymbol{x}),
\end{equation}
from which we arrive at the $n$-round success probability $P_\mathrm{succ}^{(n)}$:
\begin{equation}
    P_\mathrm{succ}^{(n)} = \frac{1}{N} \sum_{j=1}^N \tr{\rho_j\Bo^{(n)}(j)}.
\end{equation}
This expression clearly depends on the chosen post-processing.
As we want to maximize $P_\mathrm{succ}^{(n)}$, we can restrict to deterministic Markov kernels; other Markov kernels are their convex mixtures \cite{Davis61}.
The choice of the optimal post-processing will be studied in Sections. \ref{sec:main} and \ref{sec:higher}.

In performing the previously described method of state discrimination via repeated measurements, we may come across outcome arrays that do not suggest to us one particular state over another.
Let $\mathcal{S}'\subseteq\sh$ be a subset of states. 
We say that a measurement outcome array $\boldsymbol{x} \in \Omega_\Ao^n$ is  \emph{ambiguous} with respect to $\mathcal{S}'$ if $\tr{\rho\Ao^{(n)}(\boldsymbol{{x}})}=\text{const.}$ for all $\rho\in \mathcal{S}'$.
Two illustrative examples of this notion in the case of $n=1$ are the following. 
Firstly, if $\mathcal{S}'$ is the whole state space $\sh$, then a trivial observable $\To: x\mapsto p_x I$ only possesses outcomes that are ambiguous with respect to $\mathcal{S}'$. 
Secondly, when $\mathcal{S}'$ is equal to an orthonormal basis $\{\varphi_i\}$, then any sharp observable in a basis mutually unbiased to $\{\varphi_i\}$ will only possess outcomes that are ambiguous with respect to the states.

We remark that the previous framework for state discrimination via repeated measurements would work similarly had we chosen another instrument $\In$ than the L\"uders instrument.
The forms of $\Bo^{(n)}$ and $P_\mathrm{succ}^{(n)}$ obviously depend on the chosen instrument and, as said before, we do not mark the instrument simply because we stick to L\"uders instruments.

\subsection{Noisy measurement in state discrimination}

The general framework discussed in earlier subsections is applicable for any $\Ao$ and $\In$.
In this section we consider a more specific setting that will be relevant for what follows.
We begin with a collection of $N$ states $\{\rho_1, \dots,\rho_N\}$ that are perfectly distinguishable, i.e., there exists an $N$-outcome observable $\Do$ such that $\tr{\rho_i \Do(j)}=\delta_{ij}$.
This is the case if and only if they are orthogonal pure states or, more generally, mixed states with orthogonal supports.
(Note that if $N<d$, then $\Do$ is not unique.)
However, we assume that such $\Do$ is not available, and instead we use a noisy observable $\Ao$ to distinguish the states.
We assume that $\Ao$ is still reasonably good in distinguishing the states $\{\rho_1, \dots,\rho_N\}$, which we take to mean that $\tr{\rho_j\Ao(j)} \geq \tr{\rho\Ao(j)}$ and $\tr{\rho_i\Ao(j)} \leq \tr{\rho\Ao(j)}$ if $i\neq j$ for all states $\rho$.
We further make a simplifying uniformity assumption that $\tr{\rho_j\Ao(j)}$ is the same for all $j$ and similarly $\tr{\rho_i\Ao(j)}$ is the same for all $i,j$.

\begin{proposition}\label{prop:eigen}
Let
\begin{equation}\label{eq:max}
\forall j: \quad \sup_{\rho} \tr{\rho\Ao(j)} = \tr{\rho_j\Ao(j)}=\lambda
\end{equation}
and
\begin{equation}\label{eq:min}
\forall i \neq j: \quad \inf_{\rho} \tr{\rho\Ao(j)} = \tr{\rho_i\Ao(j)}=\mu \, .
\end{equation}
Then
\begin{equation}\label{eq:uniformeig-0}
        \Ao(j)\rho_i  = \begin{cases}
            \lambda\ \rho_i, \quad &i=j \, , \\
            \frac{1-\lambda}{N-1}\ \rho_i, \quad &i\neq j \, ,
        \end{cases}
    \end{equation}
    and $\lambda \geq \tfrac{1}{N}$.
\end{proposition}

\begin{proof}
It follows from \eqref{eq:max} and \eqref{eq:min} that $\lambda$ and $\mu$ are the maximal and minimal eigenvalues of $\Ao(j)$, respectively.
Then, as $\lambda$ is the maximal eigenvalue, a unit vector $\psi\in\hi$ satisfies $\ip{\psi}{\Ao(j)\psi}=\lambda$ only if $\Ao(j)\psi = \lambda \psi$.
By using the spectral decompsition of $\rho_j$ we then conclude that $\Ao(j)\rho_j = \lambda \rho_j$.
Analogous reasoning shows that $\Ao(j)\rho_i = \mu \rho_i$ for $i\neq j$.
Finally, 
\begin{equation*}
1 = \tr{\rho_i \id }=\sum_{j}\tr{\rho_i\Ao(j)} = \lambda + (N-1) \mu \, ,
\end{equation*}
from which \eqref{eq:uniformeig-0} follows.
\end{proof}

We take the setting of Prop. \ref{prop:eigen} as our starting point in the following investigations.
We remark that \eqref{eq:uniformeig-0} does not determine $\Ao$ uniquely unless $N=d$.
For example, if $\Do$ is any observable that perfectly discriminates the states, then the observable $\Ao$ given by
\begin{equation}
\Ao(j) = (\lambda - \mu )\Do(j) + \mu \id \, ,
\end{equation}
where $\mu = \frac{1-\lambda}{N-1}$, satisfies condition \eqref{eq:uniformeig-0}.

\section{Motivating qubit example}\label{sec:motiv}

In order to motivate our main result, we first provide an explicit example. 
Consider the qubit system $\C^2$ and suppose that we wish to distinguish between the two eigenstates of the $\sigma_x$ operator $P_\pm = \proj{\pm}$.
We assume that we must attempt to do so via an unsharp unbiased spin-$x$ measurement that is parametrised by $t\in[0,1]$; i.e., our observable, denoted $\Xo_t$, is given by 
\begin{equation}
\Xo_t(\pm) = \frac{I\pm t \sigma_x}{2} = \frac{1\pm t}{2} P_+ + \frac{1\mp t}{2} P_- \, .
\end{equation}
The eigenvalues of these effects are $\lambda_\pm = (1\pm t)/2$, where $\lambda_+ \geq \lambda_-$ and $\lambda_+ + \lambda_- =1$.
The success probability of distinguishing between $P_+$ and $P_-$ with $\Xo_t$ is 
\begin{equation}
    P^{(1)}_\mathrm{succ} = \frac{1}{2}\tr{P_+\Xo_t(+) + P_-\Xo_t(-)}=\lambda_+ \, .
\end{equation}

Performing the L\"uders measurement of the observable a second time leads to the sequential observable $\Xo^{(2)}_t$, with effects
\begin{equation}\label{eq:Xt2eff}
    \begin{split}
        \Xo^{(2)}_t(+,+) &= \Xo_t(+)\Xo_t(+)=\lambda_+^2 P_+ + \lambda_-^2 P_- \, ,\\
        \Xo^{(2)}_t(-,-) &= \Xo_t(-)\Xo_t(-)=\lambda_-^2 P_+ + \lambda_+^2 P_- \, ,\\
        \Xo^{(2)}_t(+,-) &= \Xo^{(2)}_t(-,+) = \Xo_t(-)\Xo_t(-)=\lambda_+\lambda_- I \, .\\
    \end{split}
\end{equation}
The first two effects can be seen to be confirmatory in nature, since the first and second measurements outcomes are in agreement, whereas the last two are ambiguous as they result in the same values for any state.
In order to assess how capable $\Xo^{(2)}_t$ is of distinguishing $P_\pm$ we must post-process the observable to create a new binary observable.
Taking the most general post-processing possible, we let $\Bo^{(2)}(+)=\sum_{i,j} w^+_{ij} \Xo^{(2)}_t (i,j)$, with weights $0\leq w^+_{ij}\leq 1$, denote the first effect of this new observable that we consider to be the ``+'' outcome.
The second effect is then $\Bo^{(2)}(-)=\id - \Bo^{(2)}(+)$.
The success probability at this stage, denoted by $P_\mathrm{succ}^{(2)}$, is given by
\begin{equation}\label{eq:qubsucc2}
    P_\mathrm{succ}^{(2)} = \frac{1}{2} \tr{P_+ \Bo^{(2)}(+)  + P_- \Bo^{(2)}(-)} = \frac{1}{2}(1+2\ \tr{P_+\Bo^{(2)}(+)} -\tr{\Bo^{(2)}(+)}),
\end{equation}
where in the last equality we have relied on $I=P_+ + P_-$.
Using the forms in Equation \eqref{eq:Xt2eff} we find that the success probability is
\begin{align*}
    P_\mathrm{succ}^{(2)} &= \frac{1}{2}\big[1 + 2 (w^{+}_{++}\lambda_+^2 +w^{+}_{--}\lambda_-^2 +(w^{+}_{+-}+w^{+}_{-+})\lambda_+\lambda_-) -(w^{+}_{++} + w^{+}_{--})(\lambda^2_+ + \lambda^2_-)- 2(w^{+}_{+-}+w^{+}_{-+})\lambda_+\lambda_-\big]\\
    &=\frac{1}{2}(1 + (w^{+}_{++}-w^{+}_{--})(\lambda_+^2 - \lambda_-^2)).
\end{align*}
Note that, as expected, neither $w^{+}_{+-}$ nor $w^{+}_{-+}$ contribute to the success probability, as ambiguous results should not be able to help us draw a conclusion about which state was measured.
Since the eigenvalues $\lambda_\pm$ sum to one, we can rewrite $\lambda_+^2 - \lambda_-^2 = 2\lambda_+ -1$.
Furthermore, since the weights are non-negative, we have $w^{+}_{++}-w^{+}_{--}\leq w^{+}_{++}\leq 1$, and so the optimal success probability arises when $w^{+}_{++}=1$ and $w^{+}_{--}=0$, which simply leads to $P_\mathrm{succ}^{(2)}=\lambda_+=P_\mathrm{succ}^{(1)}$.
In other words, performing a second measurement of $\Xo_t$ does not improve our likelihood of distinguishing between $P_+$ and $P_-$.

However, if we repeat the measurement another time then we will notice an improvement. 
The observable $\Xo^{(3)}_t$ is given by
\begin{equation}\label{eq:Xt3eff}
    \begin{split}
        \Xo_t^{(3)}(+,+,+) &= \lambda_+^3 P_+ + \lambda_-^3 P_- \, ,\\
        \Xo_t^{(3)}(+,+,-) = \Xo_t^{(3)}(+,-,+)  &= \Xo_t^{(3)}(-,+,+)= \lambda_+\lambda_-\Xo_t(+) \, , \\
        \Xo_t^{(3)}(-,+,-) = \Xo_t^{(3)}(-,+,-)  &= \Xo_t^{(3)}(+,-,-)= \lambda_+\lambda_-\Xo_t(-) \, , \\
        \Xo_t^{(3)}(-,-,-) &= \lambda_-^3 P_+ + \lambda_+^3 P_- \, .
    \end{split}
\end{equation}

With a similar calculation to before we arrive at the success probability $P_\mathrm{succ}^{(3)}=3\lambda_+^2 - 2\lambda_+^3$, which is strictly greater than $P_\mathrm{succ}^{(1)}$ for $0<t< 1$.
We have therefore seen that whilst performing the unsharp measurement twice will not provide us with an advantage in distinguishing the states, we do gain one by performing it a third time.
We should stress, however, this does not mean that the observable $\Xo^{(2)}_t$ is in general equivalent to $\Xo_t$. 
Indeed, unless $t=1$ and the observable were sharp, $\Xo^{(2)}_t$ is strictly higher in post-processing ordering \cite{MaMu90a} than $\Xo_t$ since any effect of $\Xo_t$ can be reached from post-processings of $\Xo^{(2)}_t$ but not vice versa.
It is the considered discrimination task for which $\Xo_t$ and $\Xo^{(2)}_t$ perform equally well.

\section{Binary observables}\label{sec:main}

Consider a separable (not necessarily finite) Hilbert space $\hi$ and two perfectly distinguishable states $\rho_+$ and $\rho_-$.
Let  $\Ao$ be a binary observable satisfying \eqref{eq:uniformeig-0} for these states, i.e., 
\begin{equation}
\Ao(\pm)\rho_\pm = \lambda \rho_\pm \, , \quad \Ao(\pm)\rho_\mp = (1-\lambda) \rho_\mp
\end{equation}
for some $\half \leq \lambda \leq 1$.
To simplify notation we shall adopt the convention $\Ao(+) = A$ (and hence $\Ao(-) = I-A$).
The success probability $P_\mathrm{succ}^{(1)}$ of discriminating the two states $\rho_+$ and $\rho_-$ via the observable $\Ao$ is given by
\begin{equation}
        P_\mathrm{succ}^{(1)} = \frac{1}{2}\tr{\rho_+\Ao(+) + \rho_-\Ao(-)} = \lambda \, .
\end{equation}

Since the effects $A$ and $I-A$ necessarily commute, the $n$-round observables $\Ao^{(n)}$ are of the form given in Equation \eqref{eq:commrepmeas}.
For instance, the 2-round observable $\Ao^{(2)}$ has effects
\begin{equation*}
\Ao^{(2)}(+,+)=A^2 \, , \quad \Ao^{(2)}(+,-)=\Ao^{(2)}(-,+)=A(I-A) \, , \quad \Ao^{(2)}(-,-)=(I-A)^2 \, ,
\end{equation*} 
and so forth.
Note that for the observables $\Ao^{(n)}$ the ordering of the outcomes is not reflected in the form of the effects $\Ao^{(n)}(\bx)$; instead, the only relevant fact is the total number of ``+'' or ``-'' outcomes in $\bx$. 
We can hence divide post-processing into two steps: first we group all the arrays with the same number of ``+'' outcomes, after which we study how these effects should be relabelled to form the final observable $\Bo^{(n)}$.  

Letting $p$ denote the number of ``+'' outcomes in a given $n$-length measurement array, the first step in the post-processing leads to 
\begin{equation}
    \bar{\Ao}^{(n)}(p) := \sum_{\bx \in I_{n,p}} \Ao^{(n)}(\bx) = \binom{n}{p}A^{p} (I-A)^{n-p} \, ,
\end{equation}
where $p\in \{ 0,1,\ldots,n \}$ and $I_{n,p}$ is set of all $n$-arrays containing exactly $p$  ``+'' outcomes.

In the second step we create the final binary observable $\Bo^{(n)}$, where
    \begin{align}
        \Bo^{(n)}(+)  = \sum_{p=0}^n w_p \bar{\Ao}^{(n)}(p)=\sum_{p=0}^n w_p \binom{n}{p}A^p(I-A)^{n-p},
    \end{align}
and $w_p\in \{0,1\}$ are weights determining if the outcome arrays containing $p$ ``+'' outcomes are relabelled into ``+'' or ``-''.
One might expect a `majority rule', whereby arrays with more ``+'' than ``-'' will be relabelled to ``+''. 
In the following we see that this is, indeed, the case and we further analyse the success probability.

The $n$-round success probability $P_\mathrm{succ}^{(n)}$ is 
\[
    \begin{split}
        P_\mathrm{succ}^{(n)} &= \frac{1}{2}\big(\tr{\rho_+\Bo^{(n)}(+)}+\tr{\rho_-\Bo^{(n)}(-)}\big)=\frac{1}{2}\big(1+\tr{\Bo^{(n)}(+)(\rho_+-\rho_-)}\big)\\
        &= \frac{1}{2}\left( 1 + \sum_{p=0}^n w_p \binom{n}{p}\big(\lambda^p(1-\lambda)^{n-p}-\lambda^{n-p}(1-\lambda)^p\big)\right),
    \end{split}
\]
and so in wanting to maximise the success probability we must decide the appropriate weights $w_p$.
The suitable solution for this depends on $n$, and we will therefore consider the cases for odd and even $n$ separately.
 
If we have performed an odd number of repetitions, then the sum contains an even number of terms, and so can be neatly split between $0\leq p \leq \frac{n-1}{2}$ and $\frac{n+1}{2}\leq p \leq n$.
Since $\lambda\geq1/2$, we can see that the value $\lambda^p(1-\lambda)^{n-p}-\lambda^{n-p}(1-\lambda)^p$ is negative for any value of $p$ belonging to the first half of the split, from which we conclude that its corresponding weight ought to be 0.
This means that the corresponding outcome arrays are interpreted as ``-''.
At the same time, for the second half of the split the quantity $\lambda^p(1-\lambda)^{n-p}-\lambda^{n-p}(1-\lambda)^p$ is positive and so the weight ought to be $1$.
This means that the corresponding outcome arrays are interpreted as ``+''.
Combining these pieces of information together we can conclude that in the case of an odd $n$ integer of repetitions the maximum success probability is 
\begin{equation}
    P_\mathrm{succ}^{(n)} = \frac{1}{2}\left( 1 + \sum_{p=\frac{n+1}{2}}^n \binom{n}{p} \big(\lambda^p(1-\lambda)^{n-p}-\lambda^{n-p}(1-\lambda)^p\big)\right).
\end{equation}

If we perform an even number of repetitions, then the summation for $P_\mathrm{succ}^{(n)}$ contains an odd number of terms that can be split into three groups: $0\leq p \leq \frac{n}{2}-1$, $\frac{n}{2}+1\leq p \leq n$ and $p=\frac{n}{2}$. 
By the same logic as in the odd case, $w_p=0$ for $0\leq p \leq \frac{n}{2}-1$ and $w_p = 1$ for $\frac{n}{2}+1\leq p \leq n$. 
For $p=\frac{n}{2}$ we observe that $\tr{\rho_+\bar{\Ao}^{(n)}(p)} = \tr{\rho_-\bar{\Ao}^{(n)}(p)}$, which means that the corresponding measurement outcome arrays are ambiguous with respect to $\{ \rho_+,\rho_-\}$. 
We are therefore free to use any weighting $w_{n/2}$ as it will not change the success probability, so we set $w_{n/2}=0$.
Hence, we conclude that the maximum success probability for an even $n$ is 
\begin{equation}
    P_\mathrm{succ}^{(n)} = \frac{1}{2}\left( 1 + \sum_{p=\frac{n}{2}+1}^n \binom{n}{p} \big(\lambda^p(1-\lambda)^{n-p}-\lambda^{n-p}(1-\lambda)^p\big)\right).
\end{equation}

Remarkably, according to the following theorem, an even number of repetitions provides no further improvement in success over the odd number that proceeds it:
\begin{theorem}\label{thm:binsuccprob}
    Let $n$ be an odd integer.
     The success probability of distinguishing $\rho_{+}$ and $\rho_{-}$ after $n$ repeated measurements of $\Ao$ is 
    \begin{equation}
        P_\mathrm{succ}^{(n)}=P_\mathrm{succ}^{(n+1)} = \sum_{i=\frac{n+1}{2}}^n \binom{n}{i}\binom{i-1}{\frac{n-1}{2}}(-1)^{i-\frac{n+1}{2}} \lambda^i.
    \end{equation}
\end{theorem}

\begin{figure}[t]
    \centering
    \subfloat[\label{fig:bindistcomp}]{
        \centering
        \includegraphics[width=0.5\linewidth]{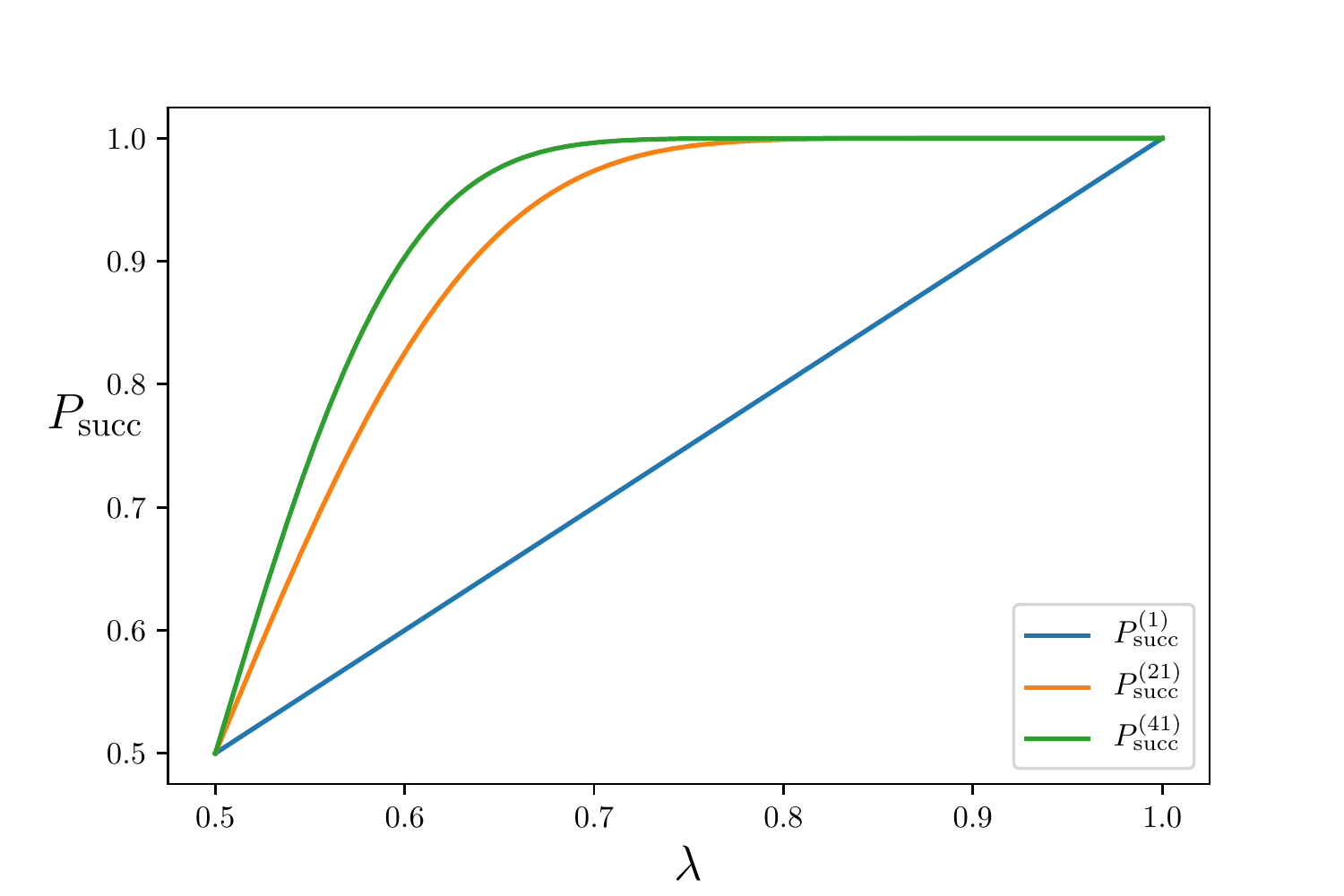}
        }
    \subfloat[\label{fig:binvalcomp}]{
        \centering
        \includegraphics[width=0.5\linewidth]{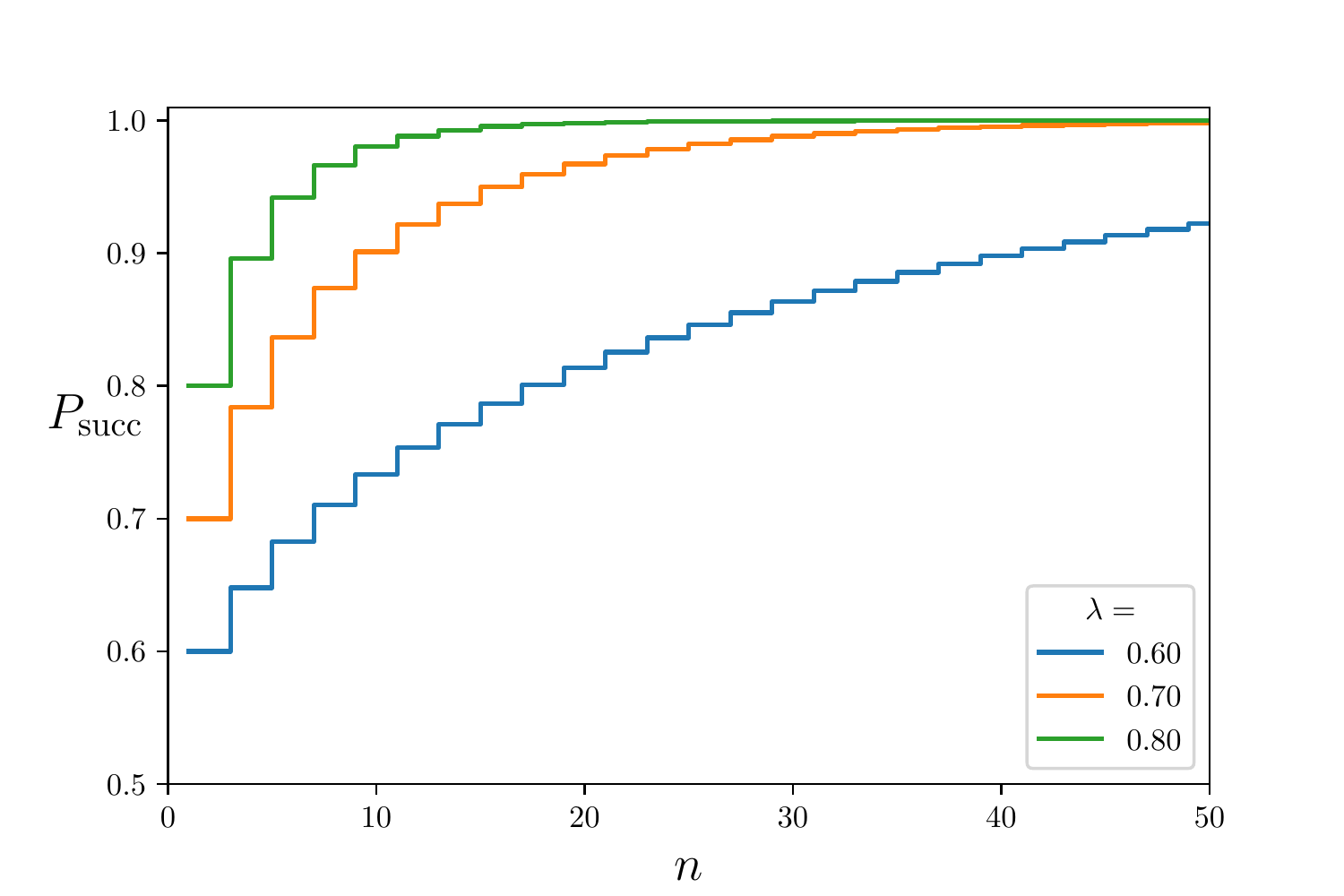}
        }
    \caption{
        The success probability of distinguishing between two eigenstates of a binary observable $\Ao$ in arbitrary dimension $d$ after $n$ measurements of $\Ao$.
        \protect\subref{fig:bindistcomp} Comparison of the success probability distribution for $n=1,21$ and $41$.
        \protect\subref{fig:binvalcomp} Comparison of the success probability for $n=1,\dots,50$ when $\lambda=0.6,0.7$ and $0.8$.
        As $n$ increases the success probability nears 1 for lower and lower values of the maximum eigenvalue $\lambda$.
        }
        \label{fig:binprob}
\end{figure}
A plot of this success probability for several values of $n$ is given in Figure \ref{fig:binprob} in terms of its overall form (Figure \ref{fig:bindistcomp}) and for particular values of $\lambda$ (Figure \ref{fig:binvalcomp}). 
As we increase $n$ we need smaller values of $\lambda$ (corresponding to noisier observables) in order to near complete success in distinguishing the states.
The proof of this result requires us to derive the success probability for both the case of $n$ and $n+1$ with odd $n$ and show that they coincide.
In order to simplify this we make use of the following lemma:
\begin{lemma}\label{lem:binomsum}
    For odd $n$ and any $\frac{n+1}{2}\leq i \leq n$,
    \begin{equation}
        \sum_{j=0}^{i-\frac{n+1}{2}}\binom{i}{j}(-1)^{j}=\binom{i-1}{\frac{n-1}{2}}(-1)^{i-\frac{n+1}{2}}.
    \end{equation}
\end{lemma}
\begin{proof}
    This is a direct consequence of the sum and found by collecting common factors:
    \begin{align*}
        \sum_{j=0}^{i-\frac{n+1}{2}}\binom{i}{j}(-1)^{j}&=1- i + \frac{i(i-1)}{2} -\frac{i(i-1)(i-2)}{3!}+\dots +\frac{i(i-1)(i-2)\dots(i-i+\frac{n-1}{2})}{(i-\frac{n+1}{2})!}(-1)^{i-\frac{n+1}{2}}\\
        &=-(i-1)\left(1 -\frac{i}{2}+ \frac{i(i-2)}{3!} -\dots +\frac{i(i-2)\dots(i-i+\frac{n-1}{2})}{(i-\frac{n+1}{2})!}(-1)^{i-\frac{n+3}{2}}\right)\\
        &=\frac{(i-1)(i-2)}{2}\left(1 - \frac{i}{3} -\dots +\frac{2i\dots(i-i+\frac{n-1}{2})}{(i-\frac{n+1}{2})!}(-1)^{i-\frac{n+5}{2}}\right)\\
        &=\frac{(-1)^{i-\frac{n+1}{2}}}{(i-\frac{n+1}{2})!}\prod_{k=1}^{i-\frac{n+1}{2}}(i-k)= (-1)^{i-\frac{n+1}{2}} \frac{(i-1)!}{(i-\frac{n+1}{2})!\frac{n-1}{2}!}=(-1)^{i-\frac{n+1}{2}} \binom{i-1}{\frac{n-1}{2}} \, .
    \end{align*}  
\end{proof}

\begin{proof}[Proof of Theorem \ref{thm:binsuccprob}]
    We first consider the case for odd $n$.
    We have already argued that the observable $\Bo^{(n)}$ that optimises the success probability must be of the form 
    \begin{align*}
    \Bo^{(n)}(+) & =\sum_{i=\frac{n+1}{2}}^n \binom{n}{i}A^i (I-A)^{n-i} \, , \quad \Bo^{(n)}(-)  = \sum_{i=\frac{n+1}{2}}^n \binom{n}{i}A^{n-i} (I-A)^{i} \, .
    \end{align*}
    Since the states $\rho_\pm$ satisfy $\tr{A\rho_+}=\tr{(I-A)\rho_-}=\lambda$, we see that $\tr{\rho_+\Bo^{(n)}(+)}=\tr{\rho_-\Bo^{(n)}(-)}$, and so the success probability reduces to 
    \[
        P_\mathrm{succ}^{(n)} = \sum_{i=\frac{n+1}{2}}^n \binom{n}{i}\lambda^i (1-\lambda)^{n-i}.
    \]
    By making use of the binomial expansion $(1-\lambda)^{n-i}=\sum_{j=0}^{n-i}\binom{n-i}{j}(-\lambda)^j$, we can rewrite the success probability as 
    \[
        \begin{split}
            P_\mathrm{succ}^{(n)}&= \sum_{i=\frac{n+1}{2}}^n \binom{n}{i}\sum_{j=0}^{n-i}\binom{n-i}{j}(-1)^j\lambda^{i+j}=\sum_{i=\frac{n+1}{2}}^n \sum_{j=0}^{n-i}\binom{n}{i+j}\binom{i+j}{j}(-1)^j\lambda^{i+j} \, .
        \end{split}
    \]
    If we let $\ell=i+j$, noting that $\ell = \frac{n+1}{2},\dots,n$, and $j=\ell-i = 0,\dots, \ell-\frac{n+1}{2}$, then we arrive at 
    \[
        P_\mathrm{succ}^{(n)}=\sum_{i=\frac{n+1}{2}}^n \binom{n}{i}\sum_{j=0}^{i-\frac{n+1}{2}}\binom{i}{j}(-1)^j \lambda^i
    \]
    which, by Lemma \ref{lem:binomsum}, is our intended result for the odd case.

    For the even case we start by recalling that the success probability does not depend on which effect contains the operators $A^{\frac{n+1}{2}}(I-A)^{\frac{n+1}{2}}$ as they do not contribute to the final success probability, and so we can choose 
    \begin{align*}
    \Bo^{(n)}(+)= \sum_{i=\frac{n+3}{2}}^{n+1}\binom{n+1}{i} A^i (I-A)^{n+1-i} \, .
    \end{align*}
    This means that the complement effect is of the form 
    \[
        \Bo^{(n)}(-)= \binom{n+1}{\frac{n+1}{2}}A^{\frac{n+1}{2}}(I-A)^{\frac{n+1}{2}} + \sum_{i=\frac{n+3}{2}}^{n+1}\binom{n+1}{i} A^{n+1-i} (I-A)^{i} \, ,
    \]
    from which we see that the success probability is
    \begin{equation}\label{eq:psn+1initial}
        P_\mathrm{succ}^{(n+1)}= \sum_{i=\frac{n+3}{2}}^{n+1}\binom{n+1}{i} \lambda^{i}(1-\lambda)^{n+1-i} +\frac{1}{2}\binom{n+1}{\frac{n+1}{2}}\lambda^{\frac{n+1}{2}}(1-\lambda)^{\frac{n+1}{2}}.
    \end{equation}
    Again making use of the binomial expansion of $(1-\lambda)^\frac{n+1}{2}$ we see that the second term in Equation \eqref{eq:psn+1initial} (omitting the factor of $\frac{1}{2}$) can be written as
    \begin{align*}
        \binom{n+1}{\frac{n+1}{2}}\lambda^{\frac{n+1}{2}}(1-\lambda)^{\frac{n+1}{2}} &=
        \sum_{i=0}^\frac{n+1}{2} \binom{n+1}{\frac{n+1}{2}}\binom{\frac{n+1}{2}}{i}(-1)^{i}\lambda^{i+\frac{n+1}{2}}\\
        &=\sum_{i=\frac{n+1}{2}}^{n+1} \binom{n+1}{\frac{n+1}{2}}\binom{\frac{n+1}{2}}{i-\frac{n+1}{2}}(-1)^{i-\frac{n+1}{2}}\lambda^{i}\\
        &=\binom{n+1}{\frac{n+1}{2}}\lambda^\frac{n+1}{2}+\sum_{i=\frac{n+3}{2}}^{n+1} \binom{n+1}{i}\binom{i}{\frac{n+1}{2}}(-1)^{i-\frac{n+1}{2}}\lambda^{i},\\
    \end{align*}
    which, when combined with the corresponding expansion for the first term
    \begin{align*}
        \sum_{i=\frac{n+3}{2}}^{n+1}\binom{n+1}{i} \lambda^{i}(1-\lambda)^{n+1-i} & =  \sum_{i=\frac{n+3}{2}}^{n+1} \sum_{j=0}^{n+1-i}\binom{n+1}{i}\binom{n+1-i}{j}(-1)^{j}\lambda^{i+j}\\
        & =  \sum_{i=\frac{n+3}{2}}^{n+1} \sum_{j=0}^{i-\frac{n+3}{2}}\binom{n+1}{i}\binom{i}{j}(-1)^{j}\lambda^{i},
    \end{align*}
    leads to the following form of the success probability:
    \begin{equation} \label{eq:psn+1interm}
        P_\mathrm{succ}^{(n+1)}= \frac{1}{2}\binom{n+1}{\frac{n+1}{2}}\lambda^\frac{n+1}{2}+ \sum_{i=\frac{n+3}{2}}^{n+1}\!\binom{n+1}{i}\!\left(\frac{1}{2}\binom{i}{\frac{n+1}{2}}(-1)^{i-\frac{n+1}{2}}+\sum_{j=0}^{i-\frac{n+3}{2}}\binom{i}{j}(-1)^{j} \right)\lambda^i.
    \end{equation}
    
    In order to reach $P_\mathrm{succ}^{(n)}$ we need to remove terms of the order $\lambda^{n+1}$, as well as binomial coefficients involving $n+1$.
    In order to resolve the first of these we note that for odd $n$
    \[
        \begin{split}
            \sum_{j=0}^{n+1}\binom{n+1}{j}(-1)^j &= \sum_{j=0}^{\frac{n-3}{2}}\binom{n+1}{j}(-1)^j + \sum_{j=\frac{n+3}{2}}^{n+1}\binom{n+1}{j}(-1)^j +\binom{n+1}{\frac{n+1}{2}}(-1)^\frac{n+1}{2}\\
            &= 2 \sum_{j=0}^{\frac{n-3}{2}}\binom{n+1}{j}(-1)^j + \binom{n+1}{\frac{n+1}{2}}(-1)^\frac{n+1}{2} = 0 \, ,
        \end{split}
    \]    
    since $\sum_{j=0}^{n+1}\binom{n+1}{j}(-1)^j= (1-1)^{n+1} = 0$.
    From this we see that the $\lambda^{n+1}$ terms cancel.
    For the binomial coefficients we begin by noting that for a given integer $i\leq n$, $\binom{n+1}{i}$ can be rewritten:
    \[
        \binom{n+1}{i} =\frac{(n+1)!}{i! (n+1-i)!}= \frac{n+1}{n+1-i}\binom{n}{i}=  \left(1 + \frac{i}{n+1-i}\right)\binom{n}{i},
    \]
    and in particular $\binom{n+1}{\frac{n+1}{2}} = 2\binom{n}{\frac{n+1}{2}}$.
    Combining these we can rewrite Equation \eqref{eq:psn+1interm} as
    \begin{align*}
        P_\mathrm{succ}^{(n+1)}&= \binom{n}{\frac{n+1}{2}}\lambda^\frac{n+1}{2} + \sum_{i=\frac{n+3}{2}}^{n}\left(1 + \frac{i}{n+1-i}\right)\binom{n}{i}\left(\frac{1}{2}\binom{i}{\frac{n+1}{2}}(-1)^{i-\frac{n+1}{2}}+\sum_{j=0}^{i-\frac{n+3}{2}}\binom{i}{j}(-1)^{j} \right)\lambda^i\\
        &=\binom{n}{\frac{n+1}{2}}\binom{\frac{n+1}{2}}{0}\lambda^\frac{n+1}{2} + \sum_{i=\frac{n+3}{2}}^{n}\left(1 + \frac{i}{n+1-i}\right)\binom{n}{i}\left(\sum_{j=0}^{i-\frac{n+1}{2}}\binom{i}{j}(-1)^{j}-\frac{1}{2}\binom{i}{\frac{n+1}{2}}(-1)^{i-\frac{n+1}{2}} \right)\lambda^i\\
        &= \sum_{i=\frac{n+1}{2}}^{n}\sum_{j=0}^{i-\frac{n+1}{2}}\binom{n}{i}\binom{i}{j}(-1)^{j}\lambda^i + \sum_{i=\frac{n+3}{2}}^{n} \frac{i}{n+1-i} \left(\binom{i-1}{\frac{n-1}{2}}(-1)^{i-\frac{n+1}{2}}-\frac{n+1}{2i}\binom{i}{\frac{n+1}{2}}(-1)^{i-\frac{n+1}{2}}\right)\lambda^i\\
        &= P_\mathrm{succ}^{(n)},
    \end{align*}
    where we have again used Lemma \ref{lem:binomsum} between the second and third line, as well as between the third and fourth.

\end{proof}

\section{Higher-outcome observables}\label{sec:higher}

In this section $\{\rho_i\}_{i=1}^N$ is a set of mutually orthogonal pure states and $\Ao$ is a commutative $N$-valued observable satisfying the conditions \eqref{eq:uniformeig-0}. 
We further assume that $1/N < \lambda < 1$ since the boundary values are not interesting.
We recall that $\Ao^{(n)}$ denotes the $n$-round observable of the form in Equation \eqref{eq:commrepmeas}, i.e., 
\begin{equation}
 \Ao^{(n)}(x_1,\ldots,x_n)=\Ao(x_1)\Ao(x_2)\cdots\Ao(x_n) \, .
\end{equation}

We start by characterising the ambiguous measurement arrays for the given task.
For any $\bx\in\Omega_\Ao^n$ and $j\in\Omega_\Ao$, we denote by 
$m(\bx,j)$ the multiplicity of $j$ in $\bx$, i.e., the number of occurences of $j$ in $\bx$.
For instance, $m((1,2,1),1)=2$, $m((1,2,1),2)=1$ and $m((1,2,1),3)=0$.

\begin{proposition}\label{prop:k-ambiglud}
Let $\bx\in\Omega_\Ao^n$.
The following are equivalent:
\begin{itemize}
\item[(i)]  $\bx$ is ambiguous with respect to $\{ \rho_{i_1},\ldots,\rho_{i_k}\}$,
\item[(ii)] $m(\bx,i_1)=m(\bx,i_2)=\cdots = m(\bx,i_k)$.
\end{itemize}
\end{proposition}

\begin{proof}
    Fix an outcome array $\bx$. 
    From Equations \eqref{eq:commrepmeas} and \eqref{eq:uniformeig-0}, we can see that for any state $\rho_j$ and result array $\bx$,
    \begin{equation}\label{eq:kambig}
        \tr{\rho_j\Ao^{(n)}(\bx)} = \lambda^{m_j} \left(\frac{1-\lambda}{N-1}\right)^{n-m_j} \, ,
    \end{equation}
    where we have denoted $m_j:=m(\bx,j)$.
    
    Let us assume that (i) holds.
    For ease of notation, and without loss of generality, let us assume that $\bx$ is ambiguous with regard to the first $k$ states, i.e., $\tr{\rho_j\Ao^{(n)}(\bx)}= \text{const.}$ for $j\in\{1,\dots,k\}$.
    By equation \eqref{eq:kambig} this means that
    \begin{equation}\label{eq:prop1pf}
        \lambda^{m_i} \left(\frac{1-\lambda}{N-1}\right)^{n-m_i} = \lambda^{m_j} \left(\frac{1-\lambda}{N-1}\right)^{n-m_j} 
    \end{equation}
    for $i,j\in \{1,\dots,k\}$.
    We make a counter assumption that (ii) does not hold, i.e., $m_i\neq m_j$ for some distinct $i,j\in \{1,\dots,k\}$.
    Without loss of generality, we can assume that $m_i>m_j$, and hence $n-m_j > n-m_i$.
    We can therefore rearrange \eqref{eq:prop1pf} to arrive at 
    \begin{equation}
        \lambda^{m_i - m_j} = \left(\frac{1-\lambda}{N-1}\right)^{m_i - m_j}.
    \end{equation}
    Since $m_i - m_j>0$, we can see that $\lambda = (1-\lambda)/(N-1)$, and so $\lambda = 1/N$. 
    However, this is in contradiction to the initial assumption $\lambda > 1/N$ and so $m_i=m_j$ for all $i,j\in\{1,\dots,k\}$.
    
    To prove that (ii) implies (i), it is sufficient to see from Equation \eqref{eq:kambig} that any set of $k$ elements of $\Omega_\Ao$ with the same multiplicity in a given outcome array $\bx$ will produce the same probabilities for their respective states.
    \end{proof}

The previous result gives another confirmation of our earlier proof of the `rule of three' in the case $N=2$. 
The situation is different for $N>2$, as in that case an outcome array $(i_1,i_2)$ is not ambiguous with respect to the full set $\{\rho_i\}_{i=1}^N$.
Hence, one might expect to gain some partial information from having a given outcome in a second measurement; that is, after obtaining (1,2) outcome array one may assume that the measured state is either $\rho_1$ or $\rho_2$ but no other.
However, as we will see, this partial information does not improve our likelihood of success in discrimination, and it is only by performing the third measurement that such results become more informative. 
The `rule of three' therefore also applies in this case.

\begin{proposition}\label{prop:2N-ary}
The success probability of distinguishing the states $\{\rho_i\}_{i=1}^N$ does not increase if a second measurement of $\Ao$ is performed, i.e., $P_\mathrm{succ}^{(2)}=P_\mathrm{succ}^{(1)}$.
\end{proposition}

\begin{proof}
For a given effect $\Ao^{(2)}(i,j)$ and state $\rho_k$ we obtain one of three possible outcomes:
    \begin{equation}
        \tr{\rho_k\Ao^{(2)}(i,j)}=\begin{cases}
            \lambda^2, \quad &i=j=k,\\
            \lambda \left(\frac{1-\lambda}{N-1} \right), \quad &i=k\neq j \ \mathrm{or} \ j=k\neq i,\\
            \left(\frac{1-\lambda}{N-1} \right)^2, \quad & i\neq k \neq j.
        \end{cases}
    \end{equation}
    We post-process $\Ao^{(2)}$ via the Markov kernel $w:(k,(i,j))\mapsto w^k_{ij}$, $i,j,k=1,\dots,N$, to form the $N$-ary observable $\Bo^{(2)}$ with effects $\Bo^{(2)}(k) = \sum_{i,j} w^k_{ij} \Ao^{(2)}(i,j)$.
    The success probability $P_\mathrm{succ}^{(2)}$ takes the form
    \begin{align}
        P_\mathrm{succ}^{(2)} &=\frac{1}{N}\tr{\rho_k \sum_{k=1}^N\Bo^{(2)}(k)}\notag\\
        &= \frac{\lambda^2}{N}\sum_{k=1}^N w^k_{kk} +\frac{\lambda(1-\lambda)}{N(N-1)} \sum_{k=1}^N\sum_{i\neq k}(w^k_{ik}+w^k_{ki}) + \frac{1}{N} \left(\frac{1-\lambda}{N-1} \right)^2\sum_{k=1}^N\sum_{i,j\neq k}w^k_{ij} \, . \label{eq:tripsuc2}
    \end{align}
    The final term in Equation \eqref{eq:tripsuc2} can be decomposed as follows:
    \begin{align}\label{eq:1-wkk}
        \sum_{k=1}^N \sum_{i,j\neq k} w^k_{ij} & = \sum_{k=1}^N\left(\sum_{i\neq k} w^k_{ii} + \sum_{i\neq k}\sum_{j\neq k,i} w^k_{ij}\right) = \sum_{k=1}^N\left(\sum_{i} w^k_{ii} -w^k_{kk} + \sum_{i\neq k}\sum_{j\neq k,i} w^k_{ij}\right) \nonumber \\
        & = \sum_{k=1}^N\sum_{i} w^k_{ii} + \sum_{k=1}^N\left(- w^k_{kk} + \sum_{i\neq k}\sum_{j\neq k,i} w^k_{ij}\right) = N + \sum_{k=1}^N\left(- w^k_{kk} + \sum_{i\neq k}\sum_{j\neq k,i} w^k_{ij}\right),
    \end{align}
    where the last equality is a consequence of the normalisation of the kernel: $\sum_k \sum_i w^k_{ii} = \sum_i \sum_k w^k_{ii} = N$.

   Next we elaborate on the middle term in Equation \eqref{eq:tripsuc2}.
       For a fixed $k$ we have
    \begin{equation}
            \sum_{i\neq k}(w^k_{ik}+w^k_{ki})+ \sum_{i\neq k} \sum_{j\neq k, i} w^k_{ij} = \sum_{i} \sum_{j\neq  i} w^k_{ij}
    \end{equation}
    and therefore
    \begin{equation}\label{eq:1-wij}
        \sum_{k=1}^N\sum_{i\neq k}(w^k_{ik}+w^k_{ki})= \sum_{k=1}^N\sum_{i} \sum_{j\neq  i} w^k_{ij} - \sum_{k=1}^N\sum_{i\neq k} \sum_{j\neq k, i} w^k_{ij} =N(N-1) - \sum_{k=1}^N\sum_{i\neq k} \sum_{j\neq k, i} w^k_{ij} \, .
    \end{equation}
    Making use of Equations \eqref{eq:1-wkk} and \eqref{eq:1-wij} we can rearrange Equation \eqref{eq:tripsuc2}:
    \begin{align}
        P_\mathrm{succ}^{(2)} &= \left(\frac{1-\lambda}{N-1} \right)^2 + \lambda(1-\lambda) + \frac{1}{N} \left(\lambda^2 - \left(\frac{1-\lambda}{N-1} \right)^2\right)\sum_{k=1}^N w^k_{kk} + \frac{1}{N} \left(\left(\frac{1-\lambda}{N-1} \right)^2 -  \frac{\lambda(1-\lambda)}{N-1} \right) \sum_{k=1}^N \sum_{i\neq k} \sum_{j\neq k,i} w^k_{ij}\notag\\
        &= \left(\frac{1-\lambda}{N-1} \right)^2 + \lambda(1-\lambda) + \frac{(N\lambda-1)(1+(N-2)\lambda)}{N(N-1)^2}\sum_{k=1}^N w^k_{kk} + \frac{(1-\lambda)(1-N\lambda)}{N(N-1)^2}\sum_{k=1}^N \sum_{i\neq k} \sum_{j\neq k,i} w^k_{ij}.\label{eq:tripsuc2red}
    \end{align}
    For the range of $\lambda$ considered, namely $\frac{1}{N} < \lambda < 1$, we have $(N\lambda-1)(1+(N-2)\lambda)>0$ whereas $(1-\lambda)(1-N\lambda) < 0$.
    Hence, in order to maximise $P_\mathrm{succ}^{(2)}$ we set the weights $w^k_{kk}=1$ for all $k$ and $w^k_{ij}=0$ whenever all indices $i,j,k$ are different.
    These choices lead to a valid Markov kernel if we further set $w^i_{jj}=0$, $w^i_{ij}=1$ and $w^j_{ij}=0$  for all $i\neq j$.
    In doing so Equation \eqref{eq:tripsuc2red} reduces to 
    \begin{equation}
        P_\mathrm{succ}^{(2)} = \left(\frac{1-\lambda}{N-1} \right)^2 + \lambda(1-\lambda)+ \lambda^2 -  \left(\frac{1-\lambda}{N-1} \right)^2= \lambda =P_\mathrm{succ}^{(1)} \, .
    \end{equation}    
\end{proof}

We proceed to higher rounds of repetitions. 

\begin{proposition}
The success probability of distinguishing the states $\{\rho_i\}_{i=1}^N$ after three and four repeated measurements of $\Ao$ is given by
    \begin{align}
        P_\mathrm{succ}^{(3)}&=\frac{1}{N-1}\lambda\big((N-2)+(N+1)\lambda - N\lambda^2\big),\label{eq:P3}\\
        P_\mathrm{succ}^{(4)}&=\frac{1}{(N-1)^2}\lambda\big((N-2)(N-3) + 3(N^2-3)\lambda + (4+7N-5N^2)\lambda^2 + 2N(N-2)\lambda^3\big).
    \end{align}
\end{proposition}

\begin{proof}
We provide the derivation of $P_\mathrm{succ}^{(3)}$; the same method applies to $P_\mathrm{succ}^{(4)}$ (with some additional care needed).
    Let $\boldsymbol{x}=(x_1,x_2,x_3)\in\{1,\dots,N\}^3$ be an array of measurement outcomes, with the corresponding effect for the observable $\Ao^{(3)}$ being $\Ao^{(3)}(\boldsymbol{x})=\Ao(x_1)\Ao(x_2)\Ao(x_3)$.
    Fix a state $\rho_j\in\{\rho_i\}_{i=1}^N$ and so, letting $m_j = m(\boldsymbol{x},j)$, the probability of obtaining outcome array $\boldsymbol{x}$ is:
    \[
        \tr{\rho_j \Ao^{(3)}(\boldsymbol{x})} = \lambda^{m_j} \left(\frac{1-\lambda}{N-1}\right)^{3-m_j}.
    \]
    We now post-process $\Ao^{(3)}$ via the deterministic kernel $w:\{1,\dots,N\} \times \{1,\dots,N\}^3 \rightarrow \{0,1\}$ to arrive at the observable $\Bo^{(3)}$ with effects $\Bo^{(3)}(y)= \sum_{\boldsymbol{x}} w(y,\boldsymbol{x})\Ao(\boldsymbol{x})$.
    Let $j\in\{1,\dots,N\}$ and $k=0,\dots,3$, then define the subsets $X^j_k = \{\boldsymbol{x}\in\{1,\dots,N\}^3\ |\ m(\boldsymbol{x},j) =k \}$.
    For $k=0,1,2,3$, these subsets have order $(N-1)^3, 3(N-1)^2, 3(N-1)$ and $1$, respectively. 
    We therefore have the following decomposition:
    \begin{equation}\label{eq:jj}
        \tr{\rho_j \Bo^{(3)}(j)} = w(j, \boldsymbol{j}) \lambda^3 + \lambda^2 \frac{1-\lambda}{N-1} \sum_{\boldsymbol{x}\in X^j_2} w(j,\boldsymbol{x}) + \lambda \left(\frac{1-\lambda}{N-1}\right)^2 \sum_{\boldsymbol{x}\in X^j_1} w(j,\boldsymbol{x}) + \left(\frac{1-\lambda}{N-1}\right)^3 \sum_{\boldsymbol{x}\in X^j_0} w(j,\boldsymbol{x}),
    \end{equation}
    where we introduce the notation $\boldsymbol{j}=(j,j,j)$. 
 
     The set $X^j_1$ has the following decomposition
    \begin{equation}\label{eq:x1decomp}
        X^j_1 = Y^j_2 \cup Y^j_{1,1},
    \end{equation}
    where $Y^j_2 = \{\boldsymbol{x}\in X^j_1\ |\ m(\boldsymbol{x},k) =2, k\neq j\}$ and $Y^j_{1,1} = \{\boldsymbol{x}\in X^j_1\ |\ m(\boldsymbol{x},k) =1 \ \mathrm{and} \ m(\boldsymbol{x},\ell) =1, k,\ell\neq j\}$.
    The subsets $Y^j_2$ and $Y^j_{1,1}$ have orders $3(N-1)$ and $3(N-1)(N-2)$, respectively, for each $j$.
    In a similar fashion, we can decompose $X^j_0$ in the following way:
    \begin{equation}\label{eq:x0decomp}
        X^j_0 = Z^j_3 \cup Z^j_{2,1} \cup Z^j_{1,1,1}
    \end{equation}
    where $Z^j_3 = \{\boldsymbol{x}\in X^j_0\ |\ m(\boldsymbol{x},k)=3, k\neq j\}$, $Z^j_{2,1} = \{\boldsymbol{x}\in X^j_0\ |\ m(\boldsymbol{x},k) =2 \ \mathrm{and} \ m(\boldsymbol{x},\ell) =1, k,\ell\neq j\}$ and $Z^j_{1,1,1} = \{\boldsymbol{x}\in X^j_0\ |\ m(\boldsymbol{x},k) =1, m(\boldsymbol{x},\ell) =1 \ \mathrm{and} \ m(\boldsymbol{x},r)=1, k,\ell,r\neq j\}$.

    Making use of Equations \eqref{eq:jj}, \eqref{eq:x1decomp} and \eqref{eq:x0decomp}, the success probability can be expressed as
    \begin{align}
        P_\mathrm{succ}^{(3)} =& \frac{\lambda^3}{N}\sum_j w(j,\boldsymbol{j}) + \frac{\lambda^2(1-\lambda)}{N(N-1)}\sum_j \sum_{\boldsymbol{x}\in X^j_2} w(j,\boldsymbol{x}) + \frac{\lambda(1-\lambda)^2}{N(N-1)^2} \sum_j\sum_{\boldsymbol{x}\in X^j_1} w(j,\boldsymbol{x}) + \frac{(1-\lambda)^3}{N(N-1)^3} \sum_j\sum_{\boldsymbol{x}\in X^j_0} w(j,\boldsymbol{x})\notag\\
        \begin{split}\label{eq:p3}
            =& \frac{\lambda^3}{N}\sum_j w(j,\boldsymbol{j}) + \frac{\lambda^2(1-\lambda)}{N(N-1)}\sum_j \sum_{\boldsymbol{x}\in X^j_2} w(j,\boldsymbol{x})+ \frac{\lambda(1-\lambda)^2}{N(N-1)^2} \sum_j\left(\sum_{\boldsymbol{x}\in Y^j_2} w(j,\boldsymbol{x})+\sum_{\boldsymbol{x}\in Y^j_{1,1}} w(j,\boldsymbol{x})\right)\\
            &+ \frac{(1-\lambda)^3}{N(N-1)^3} \sum_j \left(\sum_{\boldsymbol{x}\in Z^j_3} w(j,\boldsymbol{x})+\sum_{\boldsymbol{x}\in Z^j_{2,1}} w(j,\boldsymbol{x})+\sum_{\boldsymbol{x}\in Z^j_{1,1,1}} w(j,\boldsymbol{x})\right).
        \end{split}
    \end{align}
By noticing that $Z^j_3 = \{\boldsymbol{k}\ | \ k\neq j\}$, we can see that $\sum_{\boldsymbol{x}\in Z^j_3} w(j,\boldsymbol{x}) = \sum_k w(j,\boldsymbol{k}) - w(j,\boldsymbol{j})$, and hence
        \begin{equation}\label{eq:Z3}
            \sum_j \sum_{\boldsymbol{x}\in Z^j_3} w(j,\boldsymbol{x}) = \sum_j \sum_k w(j,\boldsymbol{k}) - \sum_j w(j,\boldsymbol{j}) = N- \sum_j w(j,\boldsymbol{j}),
        \end{equation}
        where we again utilise the normalisation of the kernel: $\sum_j \sum_k w(j,\boldsymbol{k}) = \sum_k \big(\sum_j w(j,\boldsymbol{k})\big) = N$.
        If we now denote by $S_2 = \{\boldsymbol{x} \ | \ \exists\  j, m(\boldsymbol{x},j)=2 \}$ the set of vectors containing an element of multiplicity two, then we find that for each $j$, $\sum_{\boldsymbol{x}\in Z^j_{2,1}} w(j,\boldsymbol{x}) = \sum_{\boldsymbol{x}\in S_2} w(j,\boldsymbol{x}) - \sum_{\boldsymbol{x}\in X^j_2} w(j,\boldsymbol{x}) - \sum_{\boldsymbol{x}\in Y^j_{2}} w(j,\boldsymbol{x})$. 
        In a similar way to Equation \eqref{eq:Z3} we find that
        \begin{equation}\label{eq:Z21}
            \begin{split}
                \sum_j\sum_{\boldsymbol{x}\in Z^j_{2,1}} w(j,\boldsymbol{x}) &= \sum_j \sum_{\boldsymbol{x}\in S_2} w(j,\boldsymbol{x}) - \sum_j\sum_{\boldsymbol{x}\in X^j_2} w(j,\boldsymbol{x}) - \sum_j\sum_{\boldsymbol{x}\in Y^j_{2}} w(j,\boldsymbol{x})\\
                &= 3N(N-1) - \sum_j\sum_{\boldsymbol{x}\in X^j_2} w(j,\boldsymbol{x}) - \sum_j\sum_{\boldsymbol{x}\in Y^j_{2}} w(j,\boldsymbol{x}),
            \end{split}
        \end{equation}
        where $3N(N-1)=\abs{S_2}$.
        Finally, we let $S_{1} = \{(i,j,k)\ | \ i\neq j\neq k\neq i\}$ be the set of vectors where each component has multiplicity one. 
        From this we have that $\sum_{\boldsymbol{x}\in Z^j_{1,1,1}} w(j,\boldsymbol{x}) = \sum_{\boldsymbol{x}\in S_{1}} w(j,\boldsymbol{x}) -  \sum_{\boldsymbol{x}\in Y^j_{1,1}} w(j,\boldsymbol{x})$ for each $j$, and so
        \begin{equation}\label{eq:Z111}
            \begin{split}
                \sum_j\sum_{\boldsymbol{x}\in Z^j_{1,1,1}} w(j,\boldsymbol{x}) &= \sum_j \sum_{\boldsymbol{x}\in S_1} w(j,\boldsymbol{x}) -  \sum_j\sum_{\boldsymbol{x}\in Y^j_{1,1}} w(j,\boldsymbol{x})\\
                &= N(N-1)(N-2) - \sum_j\sum_{\boldsymbol{x}\in Y^j_{1,1}} w(j,\boldsymbol{x}),
            \end{split}
        \end{equation}
        where $N(N-1)(N-2)=\abs{S_1}$.
        Inserting \Crefrange{eq:Z3}{eq:Z111} into Equation \eqref{eq:p3} we arrive at
        \begin{equation}
            \begin{split}\label{eq:P3}
                P^{(3)}_\mathrm{succ}
                =&\frac{(1-\lambda)^3}{(N-1)^3} + \frac{3(1-\lambda)^3}{(N-1)^2} + \frac{(N-2)(1-\lambda)^3}{(N-1)^2}  + \left(\frac{\lambda^3}{N} -\frac{(1-\lambda)^3}{N(N-1)^3}\right)\sum_j w(j,\boldsymbol{j})\\
                &+ \left(\frac{\lambda^2(1-\lambda)}{N(N-1)} - \frac{(1-\lambda)^3}{N(N-1)^3} \right)\sum_j \sum_{\boldsymbol{x}\in X^j_2} w(j,\boldsymbol{x})\\
                &+\left(\frac{\lambda(1-\lambda)^2}{N(N-1)^2} - \frac{(1-\lambda)^3}{N(N-1)^3}\right) \sum_j\left(\sum_{\boldsymbol{x}\in Y^j_2} w(j,\boldsymbol{x})+\sum_{\boldsymbol{x}\in Y^j_{1,1}} w(j,\boldsymbol{x})\right).
            \end{split}
        \end{equation}
        Since the $\lambda\geq 1/N$, it follows that $(1-\lambda)/(N-1)\leq 1/N\leq \lambda$, hence every factor in front of a summation in Equation \eqref{eq:P3} is positive, which suggests that each such kernel value in the summations should be equal to one.
        However, this would lead to certain outcome arrays being counted more than once: for a given array $(i,j,j)$, say, one could clearly let $w(i,(i,j,j))=1$ since $(i,j,j)\in Y^i_2$, or $w(j,(i,j,j))=1$ since $(i,j,j)\in X^j_2$, but not both because $w(i,(i,j,j))+w(j,(i,j,j))\leq 1$. Because the elements in $X^j_2$ provide a larger contribution, and $\cup_j X^j_2 = \cup_j Y^j_2$, we set $w(j,\boldsymbol{x})=1$ for $\boldsymbol{x}\in X^j_2$ and $w(j,\boldsymbol{x})=0$ for $\boldsymbol{x}\in Y^j_2$.
Similarly, for a given outcome array $(i,j,k)$ one has its inclusion three times since $(i,j,k)\in Y^i_{1,1}\cap Y^j_{1,1}\cap Y^k_{1,1}$.
        As such, the summation over the $Y^j_{1,1}$ subsets needs to be divided by three (this may be alternatively seen as randomly assigning the outcome array $(i,j,k)$ to $\Bo^{(3)}(i)$, $\Bo^{(3)}(j)$ or $\Bo^{(3)}(k)$).

        Combining these results, Equation \eqref{eq:P3} reduces to 
        \begin{equation}
            \begin{split}
                P^{(3)}_\mathrm{succ}
                =&\frac{(1-\lambda)^3}{(N-1)^3} + \frac{3(1-\lambda)^3}{(N-1)^2} + \frac{(N-2)(1-\lambda)^3}{(N-1)^2}  + \left(\frac{\lambda^3}{N} -\frac{(1-\lambda)^3}{N(N-1)^3}\right)\abs{\bigcup_j X^j_3}\\
                &+ \left(\frac{\lambda^2(1-\lambda)}{N(N-1)} - \frac{(1-\lambda)^3}{N(N-1)^3} \right)\abs{\bigcup_j X^j_2} +\frac{1}{3}\left(\frac{\lambda(1-\lambda)^2}{N(N-1)^2} - \frac{(1-\lambda)^3}{N(N-1)^3}\right) \abs{\bigcup_j Y^j_{1,1}}\\
                =&\frac{(1-\lambda)^3}{(N-1)^3} + \frac{3(1-\lambda)^3}{(N-1)^2} + \frac{(N-2)(1-\lambda)^3}{(N-1)^2}  + \left(\lambda^3 -\frac{(1-\lambda)^3}{(N-1)^3}\right)\\
                &+ 3\left( \lambda^2(1-\lambda) - \frac{(1-\lambda)^3}{(N-1)^2} \right) +(N-2)\left(\frac{\lambda(1-\lambda)^2}{N-1} - \frac{(1-\lambda)^3}{(N-1)^2}\right)\\
                =&  \lambda^3 + 3 \lambda^2(1-\lambda) +(N-2)\frac{\lambda(1-\lambda)^2}{N-1} 
                = \frac{1}{N-1}\lambda\big((N-2)+(N+1)\lambda - N\lambda^2\big).
            \end{split}
        \end{equation}
\end{proof}

    For the values of $\lambda$ considered, namely $\frac{1}{N}< \lambda < 1$, the probability $P_\mathrm{succ}^{(3)}$ is strictly larger than $P_\mathrm{succ}^{(1)}$ and $P_\mathrm{succ}^{(4)}$ strictly larger than $P_\mathrm{succ}^{(3)}$, as shown for the case of $N=10$ in Figure \ref{fig:naryprob}.
    This is in contrast to the binary case, where they coincide.

\begin{figure}[t!]
    \centering
    \includegraphics[width=0.5\textwidth]{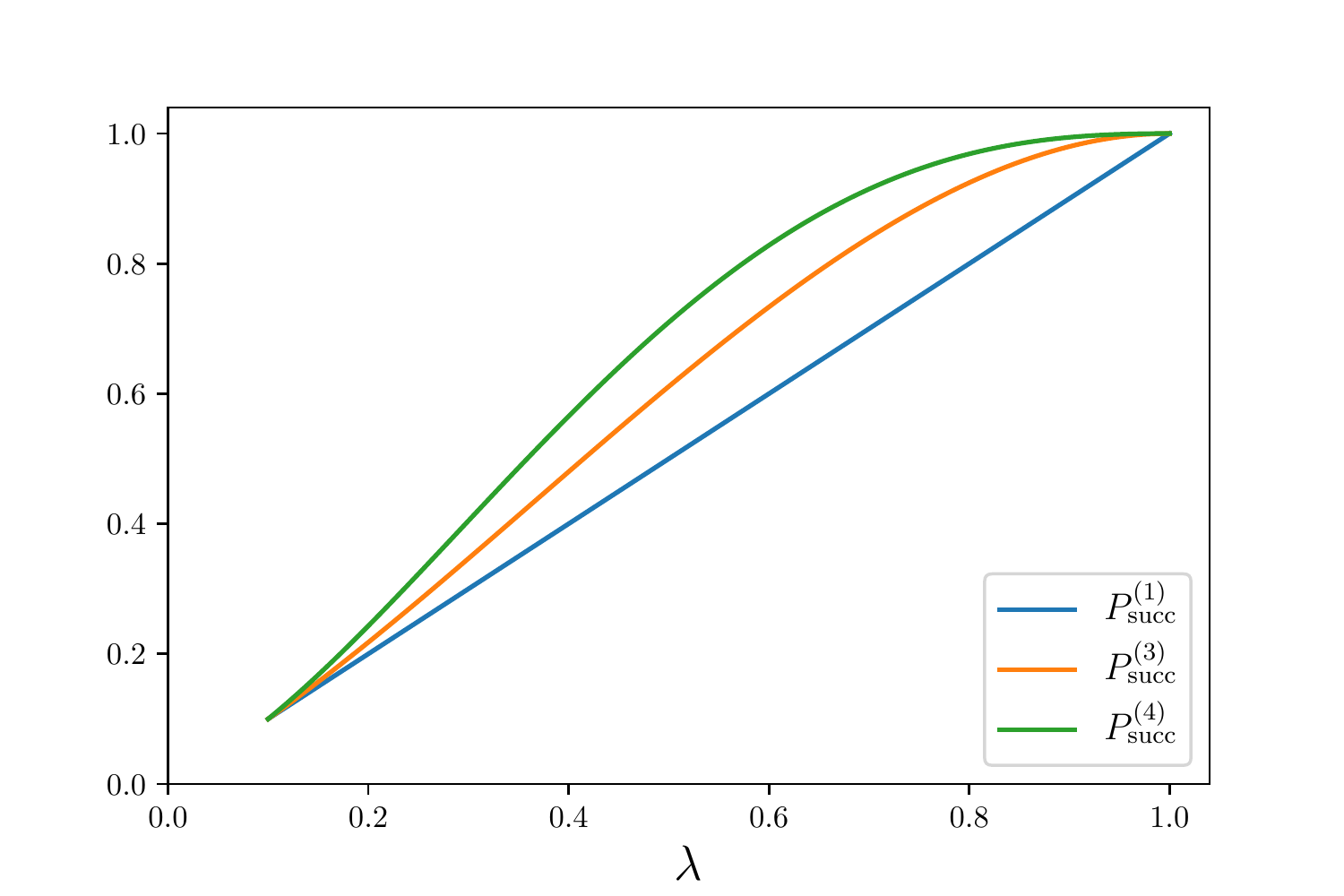}
    \caption{\label{fig:naryprob} A comparison of $P_\mathrm{succ}^{(1)}$, $P_\mathrm{succ}^{(3)}$ and $P_\mathrm{succ}^{(4)}$ for an $N=10$-outcome observable of the form given in the section, acting on a Hilbert space of arbitrary dimension $d\geq 10$.
    For any $N$-ary observable of this form, and for $\lambda \in [\frac{1}{N},1]$, the success probability of distinguishing the $N$ states after three measurements is greater than when a single measurement is performed and, contrary to the binary case, the success probability after 4 measurements is greater than after 3.}    
\end{figure}

\section{Conclusions}\label{sec:conc}

This paper has considered the task of minimum error state discrimination for a set of mutually orthogonal states, under the restriction of using one unsharp observable however many times as desired.
In the case of binary observables distinguishing two states we encountered a `rule of three', where by performing the measurement twice would provide no further advantage over a single instance, but a third measurement led to an increase for all unsharpness values $\lambda\in(\tfrac{1}{2},1)$.
As the number of repetitions increased, it was shown that the success probability would only improve when an odd number of steps were performed, which may be seen as an overcoming of ambiguous results that are present in the even cases.
While this shows that there exists benefits in performing repeated measurements, one must be cautious in the number of iterations performed, as an even number proves redundant compared to odd.

In the case of commutative $N$-valued observables with $N>2$, the rule of three was still shown to hold, but the step increase in success probability found in binary observables was no longer present. 
Since there are more possible outcomes we should not be surprised at such a result, as such ambiguity will not rise at the same points.
However, as Equation \eqref{eq:P3} shows, we do not suddenly see steps occurring every three iterations for trinary observables.
This suggests that we need not be as cautious about accidentally performing a redundant additional measurement, as in the binary case, but as is seen in Figure \ref{fig:naryprob}, the increase for $N>2$ between steps is perhaps less dramatic. 

\section*{Acknowledgements}

The authors acknowledge financial support from the Academy of Finland via the Centre of Excellence program (Project No. 312058) as well as Project No. 287750. T.B. acknowledges financial support from the Turku Collegium for Science and Medicine.

\end{document}